\documentclass[9pt,letterpaper]{amsart}
\usepackage{amssymb}
\usepackage[dvips]{epsfig}
\usepackage{graphicx}
\usepackage{color}
\usepackage{amsmath}
\usepackage{amssymb}
\usepackage{amsfonts}
\usepackage{amsthm}
\usepackage{bbm}
\usepackage{tikz}
\usepackage[colorlinks=true,linkcolor=blue,citecolor=brown]{hyperref}

\usepackage{bm}
\usepackage{stmaryrd}

\newtheorem{thm}{Theorem}[section]
\newtheorem{lem}[thm]{Lemma}

\newtheorem{cor}[thm]{Corollary}
\theoremstyle{remark}
\newtheorem{rem}{\bf Remark}[section]
\theoremstyle{definition}
\newtheorem{defn}[thm]{Definition}

\numberwithin{equation}{section}
\newtheorem{ass}{Assumption}[section]

\begin{document}
\title[The existence of full dimensional KAM tori]{The existence of full dimensional KAM tori for infinite dimensional Hamiltonian systems with long range interactions}

\author[H. Cong]{Hongzi Cong}
\address[H. Cong]{School of Mathematical Sciences, Dalian University of Technology, Dalian 116024, China}
\email{conghongzi@dlut.edu.cn}

\author[W. Ding]{Wanran Ding}
\address[W. Ding]{School of Mathematical Sciences, Dalian University of Technology, Dalian 116024, China}
\email{15998408006@163.com}

\author[Z. Zhang]{Zhihan Zhang}
\address[Z. Zhang]{School of Mathematical Sciences, Dalian University of Technology, Dalian 116024, China}
\email{dlutzzh@163.com}

\author[J. Zhao]{Juan Zhao}
\address[J. Zhao ] {School of Mathematical Sciences, Xinjiang Normal University, Urumqi 830017,China}
\email{ksc@mail.dlut.edu.cn}

\thanks{Supported by NNSFC (No.12571195, No.12371241) and NSFLP (No.2025-MS-002).}



\begin{abstract}
	We prove the existence of full dimensional KAM tori for infinite dimensional mechanical systems exhibiting long range interactions, under a Diophantine condition of Bourgain type \cite{Bour2005}, in which the radius of the invariant tori satisfies a slower decay.
\end{abstract}

\maketitle

\section{Introduction}\label{Int}

The persistence of full-dimensional invariant tori is a classical topic in KAM theory. Since the 1990s, KAM theory has been extended to infinite-dimensional settings, notably by Fr\"ohlich-Spencer-Wayne \cite{FSW1986} and P\"oschel \cite{P1990}. 
These foundational works, however, require the perturbation to exhibit sufficiently rapid spatial decay, commonly referred to as short-range interactions.

Very recently, Dolgopyat-Fayad-Paradela \cite{DFP2026}  constructed full-dimensional KAM tori for infinite-dimensional mechanical systems with long-range interactions. 
In their work, the existence of such invariant tori is established under the exponential decay condition
\begin{align*}
	I_{j} \sim e^{-\kappa j},\quad\text{with}\quad j\in\mathbb{N}, \
	\kappa>1,
\end{align*}
on the radius of the tori.

Motivated by an open problem raised by Kuksin \cite[Problem 7.1]{Kuk2004} concerning the existence of KAM tori with polynomial decay 
\begin{align*}
	I_{j} \sim |j|^{-C}, \quad \text{as } |j| \to +\infty,
\end{align*}
with some $C>0$
in the context of Hamiltonian partial differential equations, 
we investigate whether such slowly decaying invariant tori can be constructed 
for infinite-dimensional systems with long-range interactions.

\subsection{The model.}
	
	To be specific, as in \cite{DFP2026}, we consider the infinite dimensional mechanical systems with the Hamiltonian 
\begin{equation}\label{model}
	H(I,\theta)
	=
	\sum_{j\in\mathbb N} 
	\omega_{j} I_{j}
	+
	P(I,\theta),
	\quad 
	(I_{j},\theta_{j})_{j\in\mathbb N}
	\in \mathbb R^{\mathbb N} \times \mathbb T^{\mathbb N},
\end{equation}
where $\bm \omega = \left( \omega_{j} \right)_{j \in \mathbb{N}} \in  [a,b]^{\mathbb N}$ with $b>a>0$ and
\begin{align*}
	P = \sum_{i<j} P_{i,j} (\theta_i,\theta_j,I_i,I_j),
\end{align*}
which comprises long range (all-to-all) interactions.
To state our result, we first introduce some notations.

\subsection{Functional setting.}

Fix $\sigma>2$ and $w>0$. For $j \in \mathbb{N}$, we set $\lfloor j \rfloor = \max \left\{ e, |j| \right\}$.
For $\nu>0$ and a sequence $Y=(Y_{j})_{j \in \mathbb N}$, define the norms
\begin{equation}\label{eq:nu-norm}
	|Y|_\nu := \sum_{j \in \mathbb N} |Y_{j}| e^{\nu\ln^\sigma\lfloor j\rfloor},\quad \left|Y\right|_{\infty}:=	\sup_{ j \in \mathbb{N}} |Y_{j}|.
\end{equation}
The phase space is the complex Banach space
\begin{equation}\label{PhaseSpace}
	\mathfrak{H}_{\sigma,w,\infty} = \left\{
	(I,\theta) = (I_{j},\theta_{j})_{j\in\mathbb{N}}
	:
	\left|I\right|_w  < \infty,\;
\left|\Im \theta\right|_{\infty} < \infty
	\right\},
\end{equation}
with $\Re\theta_{j}\in\mathbb{T}$. 
Define
$$
J_{j}:=I_{j}-I_{j}{(0)},
\qquad
J=(J_{j})_{j\in\mathbb N}.
$$
For $r,h>0$, we define the domain
\begin{equation}\label{eq:J-localization}
	D(r,h) = \{(J,\theta)\in\mathfrak{H}_{\sigma,w,\infty} : 	\left|J\right|_{w} < r,\; \left|\Im \theta\right|_{\infty} < h\}.
\end{equation}
Furthermore, ignoring the constant terms, \eqref{model} can be rewritten as
\begin{align}\label{H}
	H = 
	\sum_{j \in \mathbb{N}}
	\omega_{j} J_{j} 
	+R(J,\theta).
\end{align}

\subsection{Weighted norm}

A Hamiltonian $R$ on $D(r,h)$ is expanded in Fourier-Taylor series around $I{(0)}$:
\begin{equation}\label{GeneralR}
	R(J,\theta) = \sum_{\bm n} R_{\bm n}(J,\theta),
\end{equation}
with
\begin{equation*}
\qquad
R_{\bm n}(J,\theta) = R(\bm n)\, J^{\bm\alpha} I(0)^{\bm\beta} e^{\mathrm{i}\langle \bm k,\theta\rangle},
\end{equation*}
where $\bm n = (\bm\alpha,\bm\beta,\bm k)\in\mathbb{N}^{\mathbb{N}}\times\mathbb{N}^{\mathbb{N}}\times\mathbb{Z}^{\mathbb{N}}$ and
$$
J^{\bm \alpha}:=\prod_{j\in\mathbb N} J_{j}^{\alpha_{j}},
\qquad
I{(0)}^{\bm \beta}:=\prod_{j\in\mathbb N} I_{j}{(0)}^{\beta_{j}},
\qquad
\langle \bm k,\theta\rangle:=\sum_{j\in\mathbb N} k_{j}\theta_{j}.
$$
For such an index $\bm n = (\bm\alpha,\bm\beta,\bm k)$, define
$$	\operatorname{supp} \bm n = \left\{ j \in \mathbb{N}
: 
\alpha_{j} + \beta_{j} + \left| k_{j} \right| \neq 0  \right\},$$
and
$$
	\overline{\bm n} = \max \left\{ j \in \operatorname{supp} \bm n \right\},
\quad
	\underline{\bm n} = \operatorname{supp} \bm n \setminus \{ \overline{\bm n}  \},
\quad
|\bm k|=\sum_{j\in\mathbb N}|k_{j}|.
$$
Similarly, for a set $A \subset \mathbb{N}$, we define
\begin{align*}
	\overline{A} = \max A
	\quad
	\text{and}
	\quad
	\underline{A} = A \setminus \{ \overline{A} \}.
\end{align*}
Inspired by \cite{DFP2026}, we define the {\it weighted norm} of the Hamiltonians as follows.
\begin{defn}\label{def:norm-plus}
	Fix $\sigma>2$ and $w>0$.
	For $\iota \in (0,1)$ and \(r,h>0\), we define
\begin{align}\label{norm}
	\interleave R \interleave_{\iota,r,h}
	=
	\sum_{j_0 \in \mathbb{N}}
	e^{w \ln^{\sigma}\lfloor j_0 \rfloor}
	\sum_{\overline{\bm n} = j_0}
	\left(
	\prod_{j \in \underline{\bm n}} 
	e^{\iota w \ln^{\sigma}\lfloor j \rfloor}
	\right)
	\left\| R_{\bm n} \right\|_{r,h}
	,
\end{align}
	where
	\begin{align*}
	\left\| R_{\bm n} \right\|_{r,h}
	=
	\left| R(\bm n) \right|
	r^{|\bm \alpha| + |\bm \beta|}
	e^{|\bm k|h}.
\end{align*}
\end{defn}

\subsection{Main result}
Before stating the main theorem, we introduce the assumption.

\begin{ass}[Diophantine]\label{ass:diophantine}
	Fix $0<d<1$.
	We assume that $\bm \omega \in [a,b]^{\mathbb N}$ is Diophantine, i.e., there exists $\eta>0$ such that
	\begin{equation}\label{eq:Diophantine}
		|\langle \bm k,\bm \omega\rangle|
		\ge
		\eta
		\prod_{j \in \underline{\bm k} }
		\left(
		\frac{1}{ 1 + k_j^2 \lfloor j \rfloor^{4} }
		\right)^{ C_0(d) }		
		,
	\end{equation}
	where
	$C_0(d) \ge \frac{4d^2 + 8d + 1}{4d(1-d)}$.
\end{ass}

\begin{thm}\label{thm:main}
	Fix $\sigma > 2$, $w,r,h>0$, and $0 < \iota < 1$. 
	Let $H$ be the Hamiltonian \eqref{H} on $D(r,h)$ with $\bm \omega = \widehat{\bm \omega}$ satisfying Assumption~\ref{ass:diophantine}. 
	Then there exists
	$$
	\epsilon_*=\epsilon_*(\sigma,w,r,h,\iota,\eta)>0
	$$
	such that, for any $0<\epsilon<\epsilon_*$ and  real-analytic $R$ on $D(r,h)$
	satisfying
	\begin{equation}\label{eq:R-small}
	\interleave R \interleave_{\iota,r,h} \le \epsilon,
	\end{equation}
	there exists a real-analytic symplectic map
$$
	\Phi: D(r/2,h/2)\to D(r,h),
$$
	close to the identity, for which
	\begin{equation}\label{eq:normal-form}
		H\circ \Phi
		=
		N_*+ R_*,
	\end{equation}
	where
$$
	N_*(J)
	=
	\sum_{j \in \mathbb N} \widehat\omega_{j} J_{j},
$$
and $R_*$ is of the form
	\begin{equation*}
	{R}_*(J,\theta)=
	\sum_{\substack{\bm n \in \mathbb{N}^{\mathbb{N}} \times  \mathbb{N}^{\mathbb{N}} \times  \mathbb{Z}^{\mathbb{N}}\\
			|\bm \alpha| \ge 2}}
	R_*(\bm n) J^{\bm \alpha} I(0)^{\bm \beta} e^{\mathrm{i}  \langle \bm k, \theta \rangle}.
\end{equation*}
	Moreover,
\begin{equation}\label{1.9}
	\interleave R_{*} \interleave_{\iota/2,r/2,h/2} \leq 2\epsilon.
\end{equation}
\end{thm}

\begin{cor}[Full-dimensional invariant tori]\label{cor:torus}
	Under the hypotheses of Theorem~\ref{thm:main}, the torus
	$$ \mathcal T_0:=\{(I,\theta):\ I=I{(0)}\} $$
	are invariant under the flow of $H\circ\Phi$, and the dynamics on $\mathcal T_0$ is the linear
	flow with frequency vector $\widehat{\bm \omega}$. 
	Consequently,
    $$ \mathcal E:=\Phi^{-1}(\mathcal T_0) $$
	is a full-dimensional invariant torus for the original Hamiltonian $H$ in
	\eqref{H}. 
\end{cor}

\subsection{Comments and outline of the proof.}
We begin with two comments, then outline the proof.
\begin{itemize}
	\item The invariant tori in Corollary~\ref{cor:torus} satisfy
	\begin{align*}
		I_j \sim e^{-\ln^{\sigma} |j|}, 
		\quad \text{with } \sigma>2.
	\end{align*}
	This decay rate lies strictly between the exponential decay $e^{-\kappa |j|}, \kappa>1 $ of \cite{DFP2026} and the polynomial decay $|j|^{-C}, C>0$.

	\item Second, we refer to Section~\ref{ME} for the precise description of the admissible frequencies. As in \cite{DFP2026}, we do not claim that our result holds for a set of large measure in $[a,b]^{\mathbb N}$ with respect to the product measure. 
	Instead, we establish the existence of invariant tori for uncountably many $\bm \omega \in [a,b]^{\mathbb N}$.
\end{itemize}

Before outlining the proof, we recall that in the PDE literature, Bourgain \cite{Bour2005} constructed KAM tori with sub-exponential decay $I_n \sim e^{-\sqrt{|n|}}$. 
This was subsequently generalized by Cong--Liu--Shi--Yuan \cite{CLSY2018} to the regime $I_n \sim e^{-|n|^\theta}$ with $\theta \in (0,1)$. 
Very recently, Cong \cite{Con2024} pushed this further to the decay $I_j \sim e^{-\ln^{\sigma} |j|}, \sigma>2$, which is the rate we adopt in the present work.

A natural question is whether this approach can be extended to infinite-dimensional mechanical systems with long-range interactions. 
The central obstacle is that 
the aforementioned works rely crucially on Bourgain's Diophantine condition
\begin{align*}
	|\langle \bm k, \bm \omega \rangle| \geq 
	\eta \prod_{j \in \operatorname{supp} k} \left(1 + k_j^2 j^4\right)^{-1}.
\end{align*}
However, as pointed out in \cite[Remark~8]{DFP2026}, the direct use of Bourgain's Diophantine condition in infinite-dimensional mechanical systems with long-range interactions is highly nontrivial, due to the lack of spatial decay with respect to the index $j$.

Indeed, the presence of the maximal index $\overline{\bm k}$
appears to be uncontrollable during the KAM iteration in the long-range setting. 
We are therefore led to impose a stronger Diophantine condition of the form \eqref{eq:Diophantine}, in which the maximal index is omitted. 
This simplification renders the iterative estimates rather standard, see Sections \ref{KAM} and \ref{Pomain}, where the analysis mostly follows \cite{Bour2005} and \cite{Con2024}.

As is well known in KAM theory, there is a delicate balance between the strength of the small-divisor condition and the size of the admissible frequency set: the stronger the Diophantine condition, the fewer frequencies satisfy it. 
Consequently, the main difficulty is shifted to the measure estimates. 
Nevertheless, following the strategy of \cite{DFP2026}, we are able to show that uncountably many frequencies still verify our condition. 
The argument hinges on two main points.
The first, adopted from \cite{DFP2026}, uses the notion of box dimension to show that the number of maximal indices is effectively finite rather than infinite (see \eqref{j0}). 
The second is an observation of the present work: since $\bm \omega$ is bounded, no resonance occurs when $k_{j_0}$ with $j_0 = \overline{\bm k}$ is sufficiently large (see \eqref{ko}).
Moreover, the dependence of the constant $C_0(d)$ in \eqref{eq:Diophantine} on the box dimension $d$ reveals that the case $d=1$  is currently beyond the scope of our approach.

\textbf{More related results.}
Recently, long-range hopping has also attracted considerable attention in the linear setting.
Shi \cite{Shi21} proved power-law localization for random operators on $\mathbb{Z}^d$ via multi-scale analysis, and later extended these ideas to almost-periodic operators by a Nash-Moser reducibility approach \cite{Shi23}. 
More recently, Shi and Wen \cite{SW25} established quantitative Green's function estimates for quasi-periodic operators with power-law long-range hopping, while Li, Wang, You and Zhou \cite{LiWYZ26} developed an operator-theoretic framework for long-range operators over general dynamical systems with analytic hopping and small potential.  
Among many other related contributions, we refer to the references cited above for a more complete picture.

\subsection*{Structure of the paper.}
The paper is organized as follows. 
In Section \ref{Int} we introduce the functional setting and state the main result. 
In section \ref{PB} we devote to the estimate of the Poisson bracket, which lies at the heart of the KAM iteration. 
We construct the iterative scheme in Section \ref{KAM} and prove the main Theorem in Section \ref{Pomain}. 
The measure estimate for the Diophantine frequencies is carried out in Section \ref{ME}. 
Finally, some auxiliary proofs are collected in Appendix \ref{Ap}.

\section{The abstract results}\label{PB}
This section establishes the fundamental estimates for the Poisson bracket, the Hamiltonian flow, and the Hamiltonian vector field in terms of the weighted norms introduced in Section~\ref{Int}. 
These estimates constitute the analytic core of the KAM
iteration in Section~\ref{KAM}.

\subsection{Poisson bracket}
\begin{defn}[Poisson bracket]\label{def:poisson}
	For two Hamiltonians $R$ and $F$ on $D(r,h)$, the Poisson bracket $\{R,F\}$ is
	defined by
	$$
	\{R,F\}
	= \sum_{\ell\in\mathbb N}
	\frac{\partial R}{\partial\theta_{\ell}}\frac{\partial F}{\partial J_{\ell}}
	-\frac{\partial R}{\partial J_{\ell}}\frac{\partial F}{\partial\theta_{\ell}}.
	$$
\end{defn}
To distinguish the Taylor-Fourier modes of $F$ from those of $R$, we write
\begin{equation}\label{F}
	F = \sum_{\bm m}F_{\bm m}
	= \sum_{\bm m}F(\bm m)\,J^{\widetilde{\bm\alpha}}\,I(0)^{\widetilde{\bm\beta}}\,
	e^{\mathrm i\langle\widetilde{\bm k},\theta\rangle},
	\qquad
	\bm m=(\widetilde{\bm\alpha},\widetilde{\bm\beta},\widetilde{\bm k}).
\end{equation}
The main estimate of this section is the following.
\begin{lem}[Poisson bracket]\label{lempb}
	Let $\sigma > 2$, $w>0$, $\iota \in (0,1)$ and $r,h,r_1,h_1,r_2,h_2>0$.
	Then for any Hamiltonians $R$ and $F$ of the form \eqref{GeneralR} and \eqref{F} respectively, we have
$$
	\interleave\{R,F\}\interleave_{\iota,r,h}
	\le 
	3
	\left(\frac{1}{r_1h_2}+\frac{1}{r_2h_1}\right) 
	\interleave R\interleave_{\iota,r+r_1,h+h_1}
	\interleave F\interleave_{\iota,r+r_2,h+h_2}.
$$
\end{lem}

\begin{proof}
	For \(\bm n,\bm m\in\mathbb N^{\mathbb N}\times\mathbb N^{\mathbb N}\times\mathbb Z^{\mathbb N}\),
	\begin{align}
		\{R_{\bm n},F_{\bm m}\}
		& = \sum_{\ell \in\operatorname{supp}\bm n\cap\operatorname{supp}\bm m}
		\frac{\partial R_{\bm n}}{\partial\theta_{\ell}}\frac{\partial F_{\bm m}}{\partial J_{\ell}}
		-\frac{\partial R_{\bm n}}{\partial J_{\ell}}\frac{\partial F_{\bm m}}{\partial\theta_{\ell}}\notag\\
		& =
		\mathrm{i} \sum_{\ell}(k_{\ell}\widetilde{\alpha}_{\ell}-\widetilde{k}_{\ell}\alpha_{\ell})\,
		R(\bm n)F(\bm m)\,J^{\bm\alpha+\widetilde{\bm\alpha}-\bm e_{\ell}}\,
		I(0)^{\bm\beta+\widetilde{\bm\beta}}\,
		e^{\mathrm i\langle\bm k+\widetilde{\bm k},\theta\rangle}.\label{A2.2}
	\end{align}
	Write
	\begin{align*}
		\left\{ R_{\bm n} , F_{\bm m} \right\} = 
		\sum_{\bm \mu \in \mathbb{N}^{\mathbb{N}} 
			\times \mathbb{N}^{\mathbb{N}}
			\times \mathbb{Z}^{\mathbb{N}}
		}
		\left\{ R_{\bm n} , F_{\bm m} \right\}_{\bm \mu}
		=
		\sum_{\bm \mu}
		\left\{ R_{\bm n} , F_{\bm m} \right\} (\bm \mu) 
		J^{\widehat{\bm \alpha}} 
		I(0)^{\widehat{\bm \beta}}
		e^{\mathrm{i} \langle \widehat{\bm k}, \theta \rangle }
	\end{align*}
	where \(\bm\mu=(\widehat{\bm\alpha},\widehat{\bm\beta},\widehat{\bm k})\). 
	By \eqref{A2.2},
	\(\{R_{\bm n},F_{\bm m}\}(\bm\mu)\neq 0\) only if there exists \(\ell \in \mathbb N\) satisfying
	\begin{align*}
		\ell \in \operatorname{supp} \bm n \cap \operatorname{supp} \bm m,
	\end{align*}
	such that
	\begin{equation}\label{Amu}
		\bm\mu = \bm \mu (\ell) 
		:=(\bm\alpha+\widetilde{\bm\alpha}-\bm e_{\ell},\ 
		\bm\beta+\widetilde{\bm\beta},\ \bm k+\widetilde{\bm k}),
	\end{equation}
	where \(\bm e_{\ell}\) denotes the unit vector with \(\ell\)-th entry equal to \(1\).
	Consequently, we have
	\begin{itemize}
		\item For $\overline{\bm n} > \overline{\bm m}$, index $\bm \mu$ satisfies
		\begin{align}\label{mu1}
			\overline{\bm \mu} = \overline{\bm n}
			\quad \text{and} \quad
			\underline{\bm \mu} \subset 
			\left( \underline{\bm n} \cup \operatorname{supp} \bm m \right).
		\end{align}	
		\item For $\overline{\bm n} < \overline{\bm m}$,  index $\bm \mu$ satisfies
		\begin{align*}
			\overline{\bm \mu} = \overline{\bm m}
			\quad \text{and} \quad
			\underline{\bm \mu} \subset 
			\left( \operatorname{supp} \bm n \cup \underline{\bm m} \right).
		\end{align*}
		\item For $\overline{\bm n} = \overline{\bm m}$, index $\bm \mu$ satisfies
		\begin{align}\label{mu3}
			\overline{\bm \mu} \le \overline{\bm n} = \overline{\bm m}
			\quad \text{and} \quad
			\underline{\bm \mu} \subset 
			\left( \underline{\bm n} \cup \underline{\bm m} \right).
		\end{align}
	\end{itemize}
	Then, we divide $\{ R,F \}$ into the following three parts.
	\begin{align*}
		\{R,F\} = f^{(1)} + f^{(2)} + f^{(3)}
		:= \sum_{\substack{\bm n,\bm m \\ \overline{\bm n} > \overline{\bm m}} } 
		\{ R_{\bm n}, F_{\bm m} \}
		+ \sum_{\substack{\bm n,\bm m \\ \overline{\bm n} < \overline{\bm m}}} 
		\{ R_{\bm n}, F_{\bm m} \}
		+ \sum_{\substack{\bm n,\bm m \\ \overline{\bm n} = \overline{\bm m}}} 
		\{ R_{\bm n}, F_{\bm m} \}.
	\end{align*}
	
	We begin by estimating $f_1$.
	In view of \eqref{norm}, \eqref{Amu} and \eqref{mu1}, we have 
	\begin{align}
		\nonumber
		\interleave f^{(1)} \interleave_{\iota,r,h}
	& = 
		\sum_{j_0 \in \mathbb{N}}
		e^{w \ln^{\sigma}\lfloor j_0 \rfloor}
		\sum_{\overline{\bm \mu} = j_0}
		\left(
		\prod_{j \in \underline{\bm \mu}} 
		e^{\iota w \ln^{\sigma}\lfloor j \rfloor}
		\right)
		\left\| f^{(1)}_{\bm \mu} \right\|_{r,h}
	\\  & \le 
		\sum_{j_0 \in \mathbb{N}}
		e^{w \ln^{\sigma}\lfloor j_0 \rfloor}
		\sum_{\overline{\bm n} = j_0}
		\sum_{\overline{\bm m} < j_0}
		\sum_{\ell \in \left( \underline{\bm n} \cap \operatorname{supp} \bm m \right) }
		\left( 
		\prod_{j \in \underline{\bm \mu (\ell)}} 
		e^{\iota w \ln^{\sigma}\lfloor j \rfloor}
		\right)
		\left\| f^{(1)}_{\bm \mu (\ell)} \right\|_{r,h}
		,
		\label{255101625}
	\end{align}
	where 
	\begin{align*}
		f^{(1)}_{\bm \mu (\ell)} = 
		\frac{\partial R_{\bm n}}{\partial \theta_{\ell}}
		\frac{\partial F_{\bm m}}{\partial J_{\ell}}
		-
		\frac{\partial R_{\bm n}}{\partial J_{\ell}}
		\frac{\partial F_{\bm m}}{\partial \theta_{\ell}}
		.
	\end{align*}
	We proceed to estimate the two factors in \eqref{255101625} separately.
	
	\noindent\textbf{(1) Estimate $\left\| f^{(1)}_{\bm \mu (\ell)} \right\|_{r,h}$: }

	Note that
	\begin{align*}
		\frac{\partial R_{\bm n}}{\partial \theta_{\ell}}
		=
		\mathrm{i}
		k_{\ell}
		R(\bm n) 
		J^{\bm \alpha} 
		I(0)^{\bm \beta}
		e^{\mathrm{i} \langle \bm k, \theta \rangle},
		\quad
		k_{\ell} \neq 0,
	\end{align*}
	we have
	\begin{align}\label{p1}
		\left\| \frac{\partial R_{\bm n}}{\partial \theta_{\ell}} \right\|_{r,h}
		=
		|k_{\ell}| \left\| R_{\bm n} \right\|_{r,h}.
	\end{align}
	From the Cauchy estimate, we have
	\begin{align}\label{p2}
			\left\| \frac{\partial R_{\bm n}}{\partial J_{\ell}} \right\|_{r,h}
		\le 
		\frac{e^{w\ln^{\sigma}\lfloor \bm l \rfloor}}{r_1}
		\left\| R_{\bm n} \right\|_{r+r_1,h}.
	\end{align}
	Combining \eqref{p1} and \eqref{p2}, we can conclude that
	\begin{align}
		\left\| f^{(1)}_{\bm \mu (\ell)} \right\|_{r,h} 
		\le 
		& \nonumber
		\left\|
		\frac{\partial R_{\bm n}}{\partial \theta_{\ell}}
		\right\|_{r,h}
		\left\|
		\frac{\partial F_{\bm m}}{\partial J_{\ell}}
		\right\|_{r,h}
		+
		\left\|
		\frac{\partial R_{\bm n}}{\partial J_{\ell}}
		\right\|_{r,h}
		\left\|
		\frac{\partial F_{\bm m}}{\partial \theta_{\ell}}
		\right\|_{r,h}
		\\  \le & \label{p3}
		\left| k_{\ell} \right|
		\frac{  e^{w\ln^{\sigma}\lfloor \ell \rfloor}}{r_2}
		\left\| R_{\bm n} \right\|_{r,h}
		\left\| F_{\bm m} \right\|_{r+r_2,h}
		+
		|\widetilde{k}_{\ell}|
		\frac{  e^{w\ln^{\sigma}\lfloor \ell \rfloor}}{r_1}
		\left\| R_{\bm n} \right\|_{r+r_1,h}
		\left\| F_{\bm m} \right\|_{r,h}.
	\end{align}
	\noindent\textbf{(2) Estimate $\prod_{j \in \underline{\bm \mu (\ell)}} 
		e^{\iota w \ln^{\sigma}\lfloor j \rfloor}$: }
	
	In view of \eqref{mu1}, since $\ell \in \left( \underline{\bm n} \cap \operatorname{supp} \bm m \right)$, we have
	\begin{align}
		\nonumber
		\prod_{j \in \underline{\bm \mu (\ell)}} 
		e^{\iota w \ln^{\sigma}\lfloor j \rfloor}
		& \le 
		\prod_{j \in \underline{\bm n} } 
		e^{\iota w \ln^{\sigma}\lfloor j \rfloor}
		\cdot
		\prod_{j \in \operatorname{supp} \bm m } 
		e^{\iota w \ln^{\sigma}\lfloor j \rfloor}
		\cdot
		e^{-\iota w \ln^{\sigma}\lfloor \ell \rfloor}
		\\ & = 
		\prod_{j \in \underline{\bm n} } 
		e^{\iota w \ln^{\sigma}\lfloor j \rfloor}
		\cdot
		e^{\iota w \ln^{\sigma}\lfloor j'_0 \rfloor}
		\prod_{j \in \underline{\bm m} } 
		e^{\iota w \ln^{\sigma}\lfloor j \rfloor}
		\cdot
		e^{-\iota w \ln^{\sigma}\lfloor \ell \rfloor},
		\quad j'_0 = \overline{\bm m}.
		\label{f11}
	\end{align}
By using \eqref{p3} and \eqref{f11}, we know the summation over $\ell$ in \eqref{255101625} can be estimated by
	\begin{align}
		\nonumber
		& \quad 
		\sum_{\ell \in \left( \underline{\bm n} \cup \operatorname{supp} \bm m \right) }
		\prod_{j \in \underline{\bm \mu (\ell)}} 
		e^{\iota w \ln^{\sigma}\lfloor j \rfloor}
		\left\| f^{(1)}_{\bm \mu (\ell)} \right\|_{r,h} 
	\\ & \le  \nonumber
		\prod_{j \in \underline{\bm n} } 
		e^{\iota w \ln^{\sigma}\lfloor j \rfloor}
		\cdot
		e^{\iota w \ln^{\sigma}\lfloor j'_0 \rfloor}
		\prod_{j \in \underline{\bm m} } 
		e^{\iota w \ln^{\sigma}\lfloor j \rfloor} 
		\sum_{\ell \in \left( \underline{\bm n} \cap \operatorname{supp} \bm m \right) }
		e^{(1-\iota) w \ln^{\sigma}\lfloor \ell \rfloor}
		\\  &  \quad  \nonumber
		\cdot
		\left(
		\frac{  \left| k_{\ell} \right| }{r_2}
		\left\| R_{\bm n} \right\|_{r,h}
		\left\| F_{\bm m} \right\|_{r+r_2,h}
		+
		\frac{ |\widetilde{k}_{\ell}| }{r_1}
		\left\| R_{\bm n} \right\|_{r+r_1,h}
		\left\| F_{\bm m} \right\|_{r,h}
		\right)
		\quad \text{with} \quad j'_0 = \overline{\bm m}
	\\ & \le  \nonumber
		\prod_{j \in \underline{\bm n} } 
		e^{\iota w \ln^{\sigma}\lfloor j \rfloor}
		\cdot
		e^{w \ln^{\sigma}\lfloor j'_0 \rfloor}
		\prod_{j \in \underline{\bm m} } 
		e^{\iota w \ln^{\sigma}\lfloor j \rfloor}
		\\ & \quad \cdot
		\left(\frac{1}{r_1h_2}+\frac{1}{r_2h_1}\right)
		\left\| R_{\bm n} \right\|_{r+r_1,h+h_1}
		\left\| F_{\bm m} \right\|_{r+r_2,h+h_2}
		\quad \text{with} \quad j'_0 = \overline{\bm m},
		\label{2.8}
	\end{align}
	where we have used the facts
	\begin{align*}
		e^{-(1-\iota)w \ln^{\sigma}\lfloor j'_0 \rfloor}
		\cdot
		e^{(1-\iota)w \ln^{\sigma}\lfloor \ell \rfloor} 
		\le 1
	\end{align*}
	and
	\begin{align*}
		\sum_{\ell} |k_{\ell}| e^{-|\bm k|h_1}
		& \le |\bm k|  e^{-|\bm k|h_1}
		\le \frac{1}{h_1}.
	\end{align*}
	Substituting \eqref{2.8} into \eqref{255101625} yields
	\begin{align}\nonumber
		\interleave f^{(1)} \interleave_{\iota,r,h}
		& \le 
		\sum_{j_0 \in \mathbb{N}}
		e^{w \ln^{\sigma}\lfloor j_0 \rfloor}
		\sum_{\overline{\bm n} = j_0}
		\sum_{j_0' < j_0}
		\sum_{\overline{\bm m} = j_0'}
		\prod_{j \in \underline{\bm n} } 
		e^{\iota w \ln^{\sigma}\lfloor j \rfloor}
		\cdot
		e^{w \ln^{\sigma}\lfloor j'_0 \rfloor}
		\prod_{j \in \underline{\bm m} } 
		e^{\iota w \ln^{\sigma}\lfloor j \rfloor}
		\\ & \quad \cdot \nonumber
		\left(\frac{1}{r_1h_2}+\frac{1}{r_2h_1}\right)
		\left\| R_{\bm n} \right\|_{r+r_1,h+h_1}
		\left\| F_{\bm m} \right\|_{r+r_2,h+h_2}
	\\ & \le \nonumber
		\left(\frac{1}{r_1h_2}+\frac{1}{r_2h_1}\right)
		\left( 
		\sum_{j_0 \in \mathbb{N}}
		e^{w \ln^{\sigma}\lfloor j_0 \rfloor}
		\sum_{\overline{\bm n} = j_0} 
		\prod_{j \in \underline{\bm n} } 
		e^{\iota w \ln^{\sigma}\lfloor j \rfloor}
		\left\| R_{\bm n} \right\|_{r+r_1,h+h_1}
		\right)
		\\ & \quad \cdot \nonumber
		\left( 
		\sum_{j_0' \in \mathbb{N}}
		e^{w \ln^{\sigma}\lfloor j'_0 \rfloor}
		\sum_{\overline{\bm m} = j_0'}
		\prod_{j \in \underline{\bm m} } 
		e^{\iota w \ln^{\sigma}\lfloor j \rfloor}
		\left\| F_{\bm m} \right\|_{r+r_2,h+h_2}
		\right)
	\\ & = \label{f1}
		\left(\frac{1}{r_1h_2}+\frac{1}{r_2h_1}\right) 
		\interleave R\interleave_{\iota,r+r_1,h+h_1}
		\interleave F\interleave_{\iota,r+r_2,h+h_2},
	\end{align}
	which finishes the estimate of $f^{(1)}$.
	
	By symmetry, repeating the above procedure, we also obtain the estimate of $f^{(2)}$:
	\begin{align}
		\label{f2}
		\interleave f^{(2)} \interleave_{\iota,r,h}
		& \le 
		\left(\frac{1}{r_1h_2}+\frac{1}{r_2h_1}\right) 
		\interleave R\interleave_{\iota,r+r_1,h+h_1}
		\interleave F\interleave_{\iota,r+r_2,h+h_2}.
	\end{align}
	
	It remains to estimate $f^{(3)}$.
	In view of \eqref{mu3}, we can obtain
	\begin{align}
		\nonumber
		\interleave f^{(3)} \interleave_{\iota,r,h}
		& = 
		\sum_{j_0' \in \mathbb{N}}
		e^{w \ln^{\sigma}\lfloor j_0' \rfloor}
		\sum_{\overline{\bm \mu} = j_0'}
		\left(
		\prod_{j \in \underline{\bm \mu}} 
		e^{\iota w \ln^{\sigma}\lfloor j \rfloor}
		\right)
		\left\| f^{(3)}_{\bm \mu} \right\|_{r,h}
	\\  & \le  \nonumber
		\sum_{j_0' \in \mathbb{N}}
		e^{w \ln^{\sigma}\lfloor j_0' \rfloor}
		\sum_{\overline{\bm n} \ge j_0'}
		\sum_{\overline{\bm m} = \overline{\bm n}}
		\sum_{\ell \in \left( \underline{\bm n} \cap \underline{\bm m} \right) }
		\prod_{j \in \underline{\bm \mu (\ell)}} 
		e^{\iota w \ln^{\sigma}\lfloor j \rfloor}
		\left\| f^{(3)}_{\bm \mu (\ell)} \right\|_{r,h}
	\\  & \le  
		\sum_{j_0 \in \mathbb{N}}
		\sum_{j_0' \le j_0}
		e^{w \ln^{\sigma}\lfloor j_0' \rfloor}
		\sum_{\overline{\bm n} = j_0}
		\sum_{\overline{\bm m} = j_0'}
		\sum_{\ell \in \left( \underline{\bm n} \cap \underline{\bm m} \right) }
		\prod_{j \in \left( \underline{\bm n} \cup \underline{\bm m} \right) } 
		e^{\iota w \ln^{\sigma}\lfloor j \rfloor}
		\left\| f^{(3)}_{\bm \mu (\ell)} \right\|_{r,h}.
		\label{f31}
	\end{align}
	Then, arguing as in the estimate of $f^{(3)}$, we obtain
	\begin{align}
		\label{f3}
		\interleave f^{(3)} \interleave_{\iota,r,h}
		& \le 
		\left(\frac{1}{r_1h_2}+\frac{1}{r_2h_1}\right) 
		\interleave R\interleave_{\iota,r+r_1,h+h_1}
		\interleave F\interleave_{\iota,r+r_2,h+h_2}.
	\end{align}
	
	Combining \eqref{f1}, \eqref{f2} and \eqref{f3}, we complete the proof of the lemma.
	
\end{proof}

\subsection{Hamiltonian flow}

Let $\Phi_F:=X_F^t|_{t=1}$, where $X_F^t$ denotes the Hamiltonian flow generated by
$F$. Iterating Lemma~\ref{lempb} along the Lie series for $R\circ\Phi_F$ yields the
following estimate.

\begin{lem}[Hamiltonian flow]\label{Hamiltonian flow}
	Let $\sigma > 2$, $w>0$, $\iota \in (0,1)$ and $r,h,r_0,h_0>0$.
	Suppose
	\begin{equation}\label{Hf}
		\frac{24e}{r_0 h_0}\,
		\interleave F\interleave_{\iota,r+r_0,h+h_0}<\frac{1}{2}.
	\end{equation}
	Then
	$$
	\interleave R\circ\Phi_F\interleave_{\iota,r,h}
	\le
	\left(1 + \frac{48e}{r_0 h_0}\,
	\interleave F\interleave_{\iota,r+r_0,h+h_0}\right)
	\interleave R\interleave_{\iota,r+r_0,h+h_0}.
	$$
\end{lem}
The proof of Lemma~\ref{Hamiltonian flow} is given in Appendix~\ref{A2}.

\subsection{Hamiltonian vector field}
For a Hamiltonian $F$, define the \emph{Hamiltonian vector field}
$$
X_F(J,\theta):=(-F_\theta,F_J),
\quad\text{with}\quad
F_{\theta} := \left\{ \frac{\partial F}{\partial \theta_{j}} \right\}_{j \in \mathbb{N}},
	\quad
F_{J} := \left\{ \frac{\partial F}{\partial J_{j}} \right\}_{j \in \mathbb{N}}.
$$
For $w>0$, set
\begin{align}\label{DefHV}
	\left\|  X_F \right\|_{w,\infty} 
	= 
	\left\|  F_{\theta} \right\|_{w} +
	\left\|  F_J \right\|_{\infty},
\end{align}
where
\begin{align*}
	\left\|  F_{\theta} \right\|_{w} = 
	\sum_{j \in \mathbb{N}} 
	\left| F_{{\theta}_{j}} \right| 
	e^{w \ln ^{\sigma} \lfloor j \rfloor}
	,\
	\left\|  F_{J} \right\|_{\infty} = 
	\sup_{j \in \mathbb{N}} 
	\left| F_{J_{j}} \right|.
\end{align*}

\begin{lem}[Hamiltonian vector field]\label{Hamiltonian vector field}
	Fix $\sigma > 2$, $w>0$ and $\iota \in (0,1)$.
	Let $F$ be a Hamiltonian on $D(r+r_0,h+h_0)$ with $r,h,r_0,h_0>0$. Then on $D(r,h)$,
	\begin{align*}
	\left\|  X_F \right\|_{w,\infty} 
	\le 
	r_0^{-1}
	\interleave F \interleave_{\iota,r+r_0,h}
	+
	h_0^{-1}
	\interleave F \interleave_{\iota,r,h+h_0}
	.
\end{align*}
\end{lem}

The proof of Lemma~\ref{Hamiltonian vector field} is given in Appendix~\ref{A3}.

\section{KAM iteration}\label{KAM}

\subsection{Derivation of homological equations}
The proof of Theorem \ref{thm:main} employs a rapidly converging Newton-type iteration to handle
Kolmogorov’s small divisor problems, involving an infinite sequence of coordinate transformations.
At the $s$-th step of the scheme, a Hamiltonian
$H_{s} = N_{s} + R_{s}$
is considered as a small perturbation of some normal form $N_{s}$ with the form of 
\begin{equation*}
	N_s=\sum_{j \in \mathbb{N}}
	\Omega_{s,j}
	J_j.
\end{equation*}
A transformation $\Phi_{s}$ is constructed such that
\begin{align*}
	H_{s}\circ \Phi_{s} = N_{s+1} + R_{s+1}
\end{align*}
yielding a new normal form $N_{s+1}$ and a smaller perturbation $R_{s+1}$. For simplicity, we drop the index $s$ from $H_{s}, N_{s}, R_{s}, \Phi_{s}, \Omega_s$ and denote the index $s+1$ by the subscript $+$.

Decompose the perturbation $R$ into three parts:
\begin{equation}\label{052301}
	R(J,\theta)=R^{(0)}(J,\theta)+R^{(1)}(J,\theta)+R^{(2)}(J,\theta),
\end{equation}
where
\begin{eqnarray*}
	{R}^{(0)}(J,\theta)
	&=&
	\sum_{\substack{\bm n \in \mathbb{N}^{\mathbb{N}} \times \mathbb{N}^{\mathbb{N}} \times \mathbb{Z}^{\mathbb{N}} \\
			|\bm \alpha| = 0}}
	R(\bm n) I(0)^{\bm \beta} e^{\mathrm{i}  \langle \bm k, \theta \rangle},
	\\
	{R}^{(1)}(J,\theta)
	&=&
	\sum_{\substack{\bm n \in \mathbb{N}^{\mathbb{N}} \times \mathbb{N}^{\mathbb{N}} \times \mathbb{Z}^{\mathbb{N}} \\
			|\bm \alpha| = 1}}
	R(\bm n) J^{\bm \alpha} I(0)^{\bm \beta} e^{\mathrm{i}  \langle \bm k, \theta \rangle},
	\\
	{R}^{(2)}(J,\theta)
	&=&
	\sum_{\substack{\bm n \in \mathbb{N}^{\mathbb{N}} \times \mathbb{N}^{\mathbb{N}} \times \mathbb{Z}^{\mathbb{N}} \\
			|\bm \alpha| \ge 2}}
	R(\bm n) J^{\bm \alpha} I(0)^{\bm \beta} e^{\mathrm{i}  \langle \bm k, \theta \rangle}.
\end{eqnarray*}
We eliminate $R^{(0)}$ and $R^{(1)}$ by the coordinate transformation $\Phi$, which is obtained as the time-1 map $X_F^{t}|_{t=1}$ of a Hamiltonian vector field $X_F$ with $F=F^{(0)}+F^{(1)}$ and 
\begin{eqnarray}\label{052001}
	&&{F}^{(i)}
	=
	\sum_{\substack{\bm n \in \mathbb{N}^{\mathbb{N}} \times \mathbb{N}^{\mathbb{N}} \times \mathbb{Z}^{\mathbb{N}} \\
			|\bm \alpha| = i, \ \bm k \neq \bm 0}}
	F(\bm n)
	J^{\bm \alpha} I(0)^{\bm \beta}
	e^{\mathrm{i}  \langle \bm k, \theta \rangle},\quad i=0,1.
\end{eqnarray}
The Hamiltonian $F$ solves the homological equation
\begin{equation}\label{4.27}
	\{N,{F}\}+R^{(0)}+R^{(1)}=\left[R^{(0)}\right]+\left[R^{(1)}\right],
\end{equation}
where
\begin{equation*}\label{051502}
	\left[R^{(i)}\right]
	=
	\sum_{\substack{\bm n \in \mathbb{N}^{\mathbb{N}} \times \mathbb{N}^{\mathbb{N}} \times \mathbb{Z}^{\mathbb{N}} \\
			|\bm \alpha| = i, \ \bm k = \bm 0}}
	R(\bm n) J^{\bm \alpha} I(0)^{\bm \beta},\quad i=0,1.
\end{equation*}
Then, the coefficients of the solution $F$ for the homological equations (\ref{4.27}) are given by
\begin{equation}\label{051304}
	F(\bm n)
	=
	\frac{R(\bm n)}
	{\sum_{j\in\mathbb{N}}  k_{j} \Omega_{j}  },
\end{equation}
and the new Hamiltonian ${H}_{+}$ takes the form
\begin{eqnarray*}
	H_{+}\nonumber&=&H\circ\Phi\\
	&=&\nonumber N + \{N,F\}  + R^{(0)} + R^{(1)}
	\\
	&&\nonumber 
	+ N\circ X_F^1 - N - \{N,F\}
	+ R\circ X_F^1 - R^{(0)} - R^{(1)}\\
	&=&N_+  + R_+,
\end{eqnarray*}
where
\begin{equation}\label{051402}
	N_+=N+\left[R^{(0)}\right]+\left[R^{(1)}\right],
\end{equation}
and
\begin{equation}\label{051403}
	R_+ = 
	N\circ X_F^1 - N - \{N,F\}
	+ R\circ X_F^1 - R^{(0)} - R^{(1)}.
\end{equation}

\subsection{KAM Iteration}\label{031501}
We now provide the precise definition of the parameters for the $s$-th
iteration.
\begin{itemize}
	\item[]$\bullet$  $\epsilon_s=\epsilon^{\left(\frac{3}{2}\right)^s}$, 
	which dominates the size of	the perturbation,
	
	\item[]$\bullet$ $\delta_{s}=\frac{1}{(s+4)\ln^2(s+4)}$,
	
	\item[]$\bullet$ $\iota_{s+1}=\iota_s-3\delta_s \iota/2$, with $\iota_0 = \iota$,			
	
	\item[]$\bullet$ $r_{s+1}=r_s-3\delta_s r/2$, with $r_0 = r$,
	
	\item[]$\bullet$ $h_{s+1}=h_s-3\delta_s h/2$, with $h_0 = h$,
	
	\item[]$\bullet$ $ d_{s+1} = d_{s} +\frac{1}{\pi^2 s^2}$ with $d_0 = 0$,
	
	\item[]$\bullet$ $D_s:=D(r_s,h_s)$,
	
	\item[]$\bullet$ $\tau_s=\epsilon_s^{0.01}$,
	
	\item[]$\bullet$ $\eta_{s+1} = \frac{1}{20} \tau_s \eta_{s}$, \text{with} $\eta_{0} = \tau_0$. 
\end{itemize}
\begin{rem}
	Since
	\begin{align*}
		\sum_{s \ge 0} \delta_s \le \int_{e}^{+\infty} \frac{1}{s \ln^2 s} \ \mathrm{ d }  s =\frac{1}{3},
	\end{align*}
	one has
	\begin{align*}
		\iota_s \ge \iota/2, \quad
		 r_s \ge r/2, \quad 
		 h_s \ge h/2, \quad
		 \text{for all} \ s \ge 0
		 .
	\end{align*}
\end{rem}

For $\tau>0$, define the complex cube
	\begin{equation*}\label{M9}
		\mathcal{B}_{\tau}\left({\omega}^*\right)
		=
		\left\{\left(\omega_{j}\right)_{j\in\mathbb{N}}\in\mathbb{C}^{\mathbb{N}}:
		|\omega-{\omega}^*|_{\infty} \leq \tau\right\}.
	\end{equation*}
Then we have the following lemma.
\begin{lem}[Iterative Lemma]{\label{IL}}
	Suppose $H_{s}=N_{s}+R_{s}$ is real analytic on $D_{s}\times\mathcal{B}_{\eta_{s}}\left(\omega_{s}^*\right)$, where
	\begin{align*}
		N_{s}=\sum_{j \in \mathbb{N}}
		 \Omega_{s,j} 
		J_{j},
	\end{align*}
	the frequency $\Omega_s=(\Omega_{s,j})_{ j\in\mathbb{N}}$ satisfies
	\begin{eqnarray}
		\label{198}&& 
		\Omega_{s,j}\left(\omega_{s}^*\right)=\widehat{\omega}_{j},\\
		\label{199}&&\left|\left|\frac{\partial {\Omega}_s(\omega)}
		{{\partial \omega }}
		-
		\operatorname{Id}
		\right|\right|_{\ell^{\infty} \to \ell^{\infty}}
		<d_s\epsilon_{0}^{\frac{1}{10}}.
	\end{eqnarray}
	Assume that $ R_{s}=R^{(0)}_{s}+R^{(1)}_{s}+R^{(2)}_{s}$ satisfies 
	\begin{eqnarray}
		\label{200}&&
		\interleave R^{(0)}_{s}\interleave_{\iota_s,r_s,h_s} \leq \epsilon_{s}:=\epsilon_{0,s},\\
		\label{201}&&
		\interleave R^{(1)}_{s}\interleave_{\iota_s,r_s,h_s} \leq \epsilon_{s}^{0.6}:=\epsilon_{1,s},\\
		\label{202}&&
		\interleave R^{(2)}_{s}\interleave_{\iota_s,r_s,h_s} \leq (1+d_s)\epsilon_0:=\epsilon_{2,s}.
	\end{eqnarray}
Let $\widehat\omega\in\mathcal D$ satisfy Assumption~\ref{ass:diophantine}.
Then
for all $\omega\in\mathcal B_{\eta_s}(\omega_s^*)$ such that
$\Omega_s(\omega)\in\mathcal B_{\tau_s}(\widehat\omega)$, there exist
\(\omega_{s+1}^*\in\mathcal B_{\eta_s}(\omega_s^*)\) and a real-analytic symplectic map
$\Phi_{s+1}:D_{s+1}\to D_s$ satisfying
	\begin{eqnarray}
		\label{203}&&
		\left\|\Phi_{s+1}-\operatorname{id}\right\|_{w,\infty}\leq \epsilon_{s}^{0.5},\\
		\label{204}&&
		\left\|D\Phi_{s+1}-\operatorname{Id}\right\|_{(w,\infty)\rightarrow(w,\infty)}\leq \epsilon_{s}^{0.5},
	\end{eqnarray}
	such that 	$H_{s+1}=H_{s}\circ\Phi_{s+1}=N_{s+1}+R_{s+1}$ satisfies
	\eqref{198}-\eqref{202} with `$s+1$' in place of `$s$', where
	$\mathcal{B}_{\eta_{s+1}}\left(\omega_{s+1}^*\right)\subset \Omega_{s}^{-1}(\mathcal{B}_{\tau_s}(\widehat{\omega}))$
	and
	\begin{align}
		&|\Omega_{s+1}(\omega)-\Omega_s(\omega)|_{\infty} \le \epsilon_s^{0.58},
		\quad\label{206}\\
		&|\omega_{s+1}^*-\omega_s^*|_{\infty} \le \epsilon_s^{0.5}.\label{205}
	\end{align}
\end{lem}

\begin{proof}
	The proof is divided into four parts.
	
\textbf{Part 1: Truncation.}

The solution \eqref{051304} is truncated to monomials satisfying
\begin{equation}\label{B1+}
	\sum_{j \in \underline{ \bm n}} \ln^{\sigma} \lfloor j \rfloor 
	< 
	\frac{\ln \epsilon_{s}^{-1}}{ w \delta_s} 
	=: B_s w^{-1} 
	\quad \text{and} \quad
	|\bm k| 
	< 
	B_s h^{-1} 
	,
	\
	s\ge1,
\end{equation}
since otherwise the term is already small enough: during the step $s\to s+1$ there are saving factors 
\begin{align*}
	\prod_{j \in \underline{\bm n}} 
	e^{-\frac{3}{2}\delta_s w \ln^{\sigma}\lfloor j \rfloor}
	\quad \text{and} \quad
	e^{-\frac{3}{2}\delta_s h |k|}.
\end{align*}

For the retained monomials, the small-divisor bound \eqref{eq:Diophantine} is controlled as follows. 
Define
\begin{align*}
	N_s^* = \frac{B_s^{\frac{\sigma-1}{\sigma}}}{\ln^{\sigma} B_s}.
\end{align*}
From \eqref{B1+}, we can conclude
\begin{align*}
	\sum_{j \in \underline{\bm k} } \ln \lfloor j \rfloor
	& \le 
	\sum_{j \in \underline{\bm k} \atop j \le N_s^* } \ln \lfloor j \rfloor
	+
	\sum_{j \in \underline{\bm k} \atop j > N_s^* } 
	\ln^{\sigma} \lfloor j \rfloor
	\cdot
	\ln^{1-\sigma} \lfloor j \rfloor
	\\ & \le 3 w^{-1/\sigma} N_s^* B_s^{1/\sigma}
	+
	w^{-1} B_s \ln^{1-\sigma} N_s^*
	\\ & 
	\le C(\sigma,w) B_s \left( \ln B_s \right)^{1-\sigma},
\end{align*}
and
\begin{align*}
	\sum_{j \in \operatorname{supp} \bm k } \ln |k_j|
	& \le C(h) B_s^{1/2} \ln B_s.
\end{align*}
Thus, we can deduce that 
\begin{align*}
	\nonumber
	\prod_{j \in \underline{\bm k} }
	\left( 1 + k_j^2  \lfloor j \rfloor^4 \right)
	\le \exp \left\{ C(\sigma,w,h) B_s \left( \ln B_s \right)^{1-\sigma} \right\}
	\ll \epsilon_s^{-1},
\end{align*}
where $C(\sigma,w,h)$ is a constant depending on $\sigma$, $w$, $h$, and we have used the fact that $\sigma>2$.
Consequently, we have 
\begin{align}
	\label{epsilon-0.01}
	\eta^{-1} \prod_{j \in \underline{\bm k} }
	\left( 1 + k_j^2  \lfloor j \rfloor^4 \right)^{C_0(d)}
	\le \epsilon_s^{-0.01}. 
\end{align}

\textbf{Part 2. Estimate on the solutions of homological equation.}

We prove
\begin{align}
	\interleave F_s^{(0)}\interleave_{\iota_s,r_s,h_s} &\le \epsilon_s^{0.98},\label{022410.}\\
	\interleave F_s^{(1)}\interleave_{\iota_s,r_s,h_s} &\le \epsilon_s^{0.58}.\label{022410..}
\end{align}
	In view of \eqref{051304}, we have
	\begin{align*}
		\interleave F_s^{(0)}\interleave_{\iota_s,r_s,h_s}
		&= 
		\sum_{j_0\in\mathbb{N}} 
		e^{w\ln^\sigma\lfloor j_0 \rfloor}
		\sum_{\overline{\bm n} = j_0}
		\left(
		\prod_{j \in \underline{\bm n}} 
		e^{\iota w \ln^{\sigma}\lfloor j \rfloor}
		\right)
		\left\|
		\frac{R_{s,\bm n}^{(0)}}
		{\sum_{j}k_{j}\Omega_{s,j}}
		\right\|_{r_s,h_s}
		\\
		&\le \interleave R_s^{(0)}\interleave_{\iota_s,r_s,h_s}\,\epsilon_s^{-0.01}
		\qquad\text{(by \eqref{eq:Diophantine} and \eqref{epsilon-0.01})}\\
		&\le \epsilon_s^{0.98},
	\end{align*}
	where we used \eqref{200}. The proof of \eqref{022410..} is analogous.
	
	Moreover, combining \eqref{022410.}, \eqref{022410..} and Lemma \ref{Hamiltonian vector field},
	\begin{equation*}
		\|X_{F_s}\|_{w,\infty}\leq\epsilon_{s}^{0.57},
	\end{equation*}
which implies \eqref{203} and \eqref{204}.

\textbf{Part 3. Estimate on the remainder $R_{s+1}$.}

Recalling \eqref{051403}, the remainder $R_{s+1}$ is defined by
\[
R_{s+1} = N_s\circ X_{F_s}^1 - N_s - \{N_s,F_s\}
+ R_s\circ X_{F_s}^1 - R_s^{(0)} - R_s^{(1)}.
\]
Applying Lemma~\ref{lempb} and Lemma~\ref{Hamiltonian flow}, together with \eqref{200}-\eqref{202} and
\eqref{022410.}-\eqref{022410..}, 
one verifies that \eqref{200}-\eqref{202} hold at stage \(s+1\).

	\textbf{ Part 4. Estimate on the frequency shift.}

	Recall \eqref{051402}, we know that
	\begin{align*}
		N_{s+1}=N_s+\left[R_s^{(0)}\right]+\left[R_s^{(1)}\right].
	\end{align*}
	Since $\left[R_s^{(0)}\right]$ is a constant which does not affect the Hamiltonian vector field,
	we can denote by
	\begin{align*}
		\Omega_{s+1,j} = \Omega_{s,j} + 
		\sum_{
			\substack{
				\bm n \in 
				\mathbb{N}^{\mathbb{N}} \times
				\mathbb{N}^{\mathbb{N}} \times
				\mathbb{Z}^{\mathbb{N}} 
				\\
				\bm \alpha = \bm e_{j}, \ \bm k = \bm 0
			}
		}
		R_s(\bm n) I(0)^{\bm \beta},
	\end{align*}
	where $\bm e_{j}$ is the unit vector with the $j$-th element equals to $1$ and the others are $0$, the last term is the so-called frequency shift.

 $\bullet$	\emph{Proof of \eqref{206}.}
	On \(D_{s+1}\),
	\begin{align*}
		\sup_{j\in \mathbb{N}}
		\left|
		\sum_{
			\substack{
				\bm n \in 
				\mathbb{N}^{\mathbb{N}} \times
				\mathbb{N}^{\mathbb{N}} \times
				\mathbb{Z}^{\mathbb{N}} 
				\\
				\bm \alpha = \bm e_{j}, \ \bm k = \bm 0
			}
		}
		R_s(\bm n) I(0)^{\bm \beta}		
		\right|
			&  \le
		\sup_{j\in \mathbb{N}}  
		\sum_{\bm n \atop \operatorname{supp} \bm n \ni j} 
		\left\| 
		\frac{\partial R^{(1)}_{s,\bm n} }{\partial J_{j}} \right\|_{r_{s+1},0}
		\\
		& \le
		\frac{2}{3\delta_sr}
		\sum_{j_0 \in \mathbb{N}}
		e^{w\ln^{\sigma}\lfloor j_0 \rfloor} 
		\sum_{\overline{\bm n} = j_0}
		\left\| R^{(1)}_{s,\bm n} \right\|_{r_s,0}  
		\\
		& \quad (\text{by the Cauchy estimate})
		\\
		& \le 
		\frac{2}{3\delta_sr}
		\interleave R^{(1)}_{s}\interleave_{\iota_s,r_s,h_s}
		\le \epsilon_s^{0.58},
	\end{align*}
	where we have used \eqref{201}.
	Thus \eqref{206} holds.
	
	$\bullet$ \emph{Proof of \eqref{199} for \(s+1\).}
	By Cauchy's estimate on \(\mathcal B_{\tau_s\eta_s}(\omega_s^*)\),
	\begin{equation}\label{620}
		\sup_{j'}\sum_{j}
		\left|\frac{\partial\Omega_{s+1, j'}}{\partial\omega_{j}}
		-\frac{\partial\Omega_{s,j'}}{\partial\omega_{j}}\right|
		\le 
		\frac{|\Omega_{s+1}-\Omega_s|_{\infty}}{\tau_s\eta_s}
		\le 
		\epsilon_s^{0.55},
	\end{equation}
	where we used \eqref{206}.
	 On
	\(\mathcal B_{\frac{1}{10}\tau_s\eta_s}(\omega_s^*)\), combining \eqref{199} with
	\eqref{620} gives
	\begin{align*}
		&\quad 
		\sup_{j'}\sum_{j}
		\left| 
		\frac{\partial \Omega_{s+1,j'}}{\partial \omega_{j}} 
		-
		\delta_{j j'}
		\right|
		\\ & \le
		\sup_{j'}\sum_{j}
		\left| 
		\frac{\partial \Omega_{s+1,j'}}{\partial \omega_{j}} 
		-
		\frac{\partial \Omega_{s, j'}}{\partial \omega_{j}}
		\right|
		+
		\sup_{j'}\sum_{j}
		\left| 
		\frac{\partial \Omega_{s,j'}}{\partial \omega_{j}} 
		-
		\delta_{j j'}
		\right|
		\\ & \le 
		\epsilon_{s}^{0.55} + d_s \epsilon_0^{\frac{1}{10}}
		\\ & 
		< d_{s+1}\epsilon_0^{\frac{1}{10}},
	\end{align*}
	which verifies \eqref{199} for $s+1$.
	
$\bullet$	\emph{Proof of \eqref{198} for \(s+1\).}
	Consider the functional equation
	\begin{equation}\label{M21}
		\Omega_{s+1}(X) = \widehat\omega,
		\qquad X\in\mathcal B_{\frac{1}{10}\tau_s\eta_s}(\omega_s^*).
	\end{equation}
	By \eqref{199} and the inverse function theorem, \eqref{M21} has a solution
	\(\omega_{s+1}^*\).
	
$\bullet$ \emph{Proof of \eqref{205}.}
Write
\[
\omega_{s+1}^*-\omega_s^*
= (\mathrm{I}-\Omega_{s+1})(\omega_{s+1}^*)
- (\mathrm{I}-\Omega_{s+1})(\omega_s^*)
+ (\Omega_s-\Omega_{s+1})(\omega_s^*).
\]
By \eqref{206} and \eqref{620}, we have
\[
|\omega_{s+1}^*-\omega_s^*|_{\infty}\le\epsilon_s^{0.5}\ll\tau_s\eta_s,
\]
which completes the proof.
\end{proof}

\section{The proof of main result }\label{Pomain}
\begin{proof}[Proof of the Theorem \ref{thm:main}]
		Recall that the Hamiltonian $H$ takes the form of \eqref{H}.
	To apply Lemma~\ref{IL} with $s=0$, set
	$$
	H_0=H,\quad \omega_0^*=\widehat\omega,\quad \Omega_0=\operatorname{id},\quad \epsilon_0=\epsilon.
	$$
	Then \eqref{198}-\eqref{202} hold with $s=0$. By induction, Lemma~\ref{IL} yields a
	decreasing sequence of domains $D_s\times\mathcal B_{\eta_s}(\omega_s^*)$ and a
	sequence of transformations
	$$
	\Phi^s := \Phi_1\circ\cdots\circ\Phi_s
   $$
	such that $H\circ\Phi^s=N_s+R_s$ for $s\ge 1$, with
	\eqref{203}-\eqref{205} holding at each step.
	
	Consequently, $\omega_s^*$ converges to
	a limit $\omega^*$ with
	$$
	|\omega^*-\widehat\omega|_{\infty}\le\sum_{s=0}^\infty\epsilon_s^{0.5}<\epsilon^{0.4},
	$$
	and $\Phi^s$ converges uniformly on $D(r/2,h/2)\times\{\omega^*\}$ to $\Phi$ with
	$$
	\|\Phi-\operatorname{id}\|_{w,\infty}\le\epsilon^{0.4},
	\qquad
	\|D\Phi-\operatorname{Id}\|_{(w,\infty)\to(w,\infty)}\le\epsilon^{0.4}.
	$$
	Letting $s\to\infty$ completes the proof.
\end{proof}

\section{Measure estimate}\label{ME}

In this section we prove that the set of frequencies satisfying the Diophantine condition of Assumption~\ref{ass:diophantine} is uncountable.
The proof follows the strategy of \cite{DFP2026}, with the necessary adaptations to the present setting, where we handle a slight variant of the Diophantine condition originally introduced by Bourgain \cite{Bour2005}.

\subsection{Probability measure and main theorem}
Take any $0<d<1$ and denote by $\mathcal{C}(d) \subset [a,b]^\mathbb{N}$ the set of frequencies for which $\operatorname{dim}_{\text{box}} \left( \widetilde{\bm \omega} \right) = d$.
It is around these frequencies that we will construct
a suitable probability measure for which Diophantine vectors occupy a large measure.
To be specific, given any $\ell > 0$ and $\widetilde{\bm \omega} \in \mathcal{S}(d)$, we consider the set
\begin{align*}
	\Lambda \left( \widetilde{\bm \omega} \right)
	=
	\prod_{j \in \mathbb{N}} \Lambda_j
	=
	\prod_{j \in \mathbb{N}}
	\left[ \widetilde{\omega}_{j} - \frac{1}{j^{\frac{1}{d}}},
	\, \widetilde{\omega}_{j} + \frac{1}{j^{\frac{1}{d}}} \right]
	,
\end{align*}
endowed with the product topology.
The probability measure is defined by
\begin{align*}
	\operatorname{meas} \left( \prod_{j \in \mathbb{N}} \Lambda_j \right)
	=
	\prod_{j \in \mathbb{N}}
	\frac{\operatorname{Leb(\Lambda_j)}}{2 j^{-\frac{1}{d}}}
	.
\end{align*}
Recall that the Diophantine condition is given by \eqref{eq:Diophantine}. Define the resonant set
$$
\mathfrak R := 
\bigcup_{0\neq\bm k\in\mathbb Z^{\mathbb N}} 
\mathfrak R(\bm k),
$$
where
\begin{align}\label{Rk}
	\mathfrak R(\bm k) = \left\{\bm \omega \in \Lambda\left( \widetilde{\bm \omega} \right):
	|\langle \bm k,\bm \omega \rangle|
	<
	\eta
	\prod_{j \in \underline{\bm k} }
	\left(
	\frac{1}{ 1 + k_j^2 \lfloor j \rfloor^{4} }
	\right)^{ C_0(d) }
	\right\}
\end{align}
with $C_0(d) \ge \frac{4d^2 + 8d + 1}{4d(1-d)}$.
The main theorem of this section is as follows.
\begin{thm}\label{thm:measure}
	Let $\widetilde{\bm \omega} \in [a,b]^{\mathbb{N}}$ have box dimension $0<d<1$. For $\Lambda(\bm \omega)$ and $\mathfrak R$ defined above, we have
	\begin{align*}
		\operatorname{meas} \mathfrak R \le  C(a,b,d)  \eta^{1-d},
	\end{align*} 
	where $C(a,b,d)$ is a constant depending on $a$, $b$ and $d$.
\end{thm}
Before proving this theorem, we first make some preparations.
To be specific, following the strategy in \cite{DFP2026}, we analyze the structure of the resonant set $\mathfrak R$ by using the definition of box dimension.

\subsection{Preliminary analysis}
First, we split $\bm k$  into $\bm k = (\bm k_*,k_{j_0})$, where $j_0 = \overline{\bm k}$, $k_{j_0}\in \mathbb{Z}$ and $\bm k_* \in \mathbb{Z}^{\mathbb{N}\setminus \{j_0\}}$.
Denote by $\Pi$ the projection $\bm k \mapsto \bm k_* $.
We can rewrite the union in the resonant set $\mathfrak{R}$ as
\begin{align}\label{Rj0}
	\mathfrak R = 
	\bigcup_{\bm k_* \in \mathbb{Z}^{\mathbb{N}}} 
	\bigcup_{j_0 > \overline{\bm k_*}}
	\bigcup_{\substack{\overline{\bm k}=j_0 \\ \Pi (\bm k) = \bm k_* } }
	\mathfrak R (\bm k).    	
\end{align}
Then we establish the following lemma 
to replace the infinite union over $j_0$ in \eqref{Rj0} by a finite one.
\begin{lem}\label{Lem5.1}
	Given $\bm k_* \in \mathbb{Z}^{\mathbb{N}}$, $0<d<1$, and we define
	\begin{align}\label{mfJ}
		\mathfrak{J}_{\bm k_*} = \min \left\{ j \in \mathbb{N}: B_{\bm k_*} j^{-\frac{1}{d}} \le 
		\eta
		\mathfrak{B}_{\bm k_*}^{ C_0(d) }  \right\},
	\end{align}
	where
	\begin{align*}
		B_{\bm k_*} =  \frac{1}{a} + \frac{b}{a} 
		\sum_{j \in \operatorname{supp} \bm k_*} \left|{k_*}_j\right|
		\quad \text{and} \quad
		\mathfrak{B}_{\bm k_*} = 
		\prod_{j \in \operatorname{supp} \bm k_* }
		\frac{1}{ 1 + k_j^2 \lfloor j \rfloor^{4} }.
	\end{align*}
	There exists a subset $\mathcal{I}_{\bm k_*} \subset \mathbb{N}$ with
	\begin{align*}
		\operatorname{card} \left(\mathcal{I}_{\bm k_*}\right) \lesssim 
		\mathfrak{J}_{\bm k_*},
	\end{align*}
	such that
	\begin{align}\label{j0}
		\bigcup_{j_0 > \overline{\bm k_*}}
		\bigcup_{\substack{\overline{\bm k}=j_0 \\ \Pi (\bm k) = \bm k_* } }
		\mathfrak R (\bm k)
		\subset
		\left( \bigcup_{j_0 \le \mathfrak{J}_{\bm k_*} }
		\bigcup_{\substack{\overline{\bm k}=j_0 \\ \Pi (\bm k) = \bm k_* } }
		\mathfrak R (\bm k) \right)
		\bigcup
		\left( \bigcup_{j_0 \in \mathcal{I}_{\bm k_*}}
		\bigcup_{\substack{\overline{\bm k}=j_0 \\ \Pi (\bm k) = \bm k_* } }
		\widehat{\mathfrak R }(\bm k) \right)
		,
	\end{align}
	where the set $\widehat{\mathfrak R }(\bm k)$ is defined as in \eqref{Rk} but with the right-hand side multiplied by $5$.
\end{lem}
\begin{proof}
	Since $\widetilde{\bm \omega} \in \mathcal{C}(d)$, there exists $C > 0$ such that $\{\widetilde{\omega}_{j}\}_{j \geqslant \mathfrak{J}_{\bm k_*}}$ can be covered with less than $C \mathfrak{J}_{\bm k_*}$ balls of size $\mathfrak{J}_{\bm k_*}^{-\frac{1}{d}}$. 
	This implies that we can find a subset $\mathcal{I}_{\bm k_*} \subset \mathbb{N} \cap \{j \geqslant \mathfrak{J}_{\bm k_*}\}$ of cardinality less than $C \mathfrak{J}_{\bm k_*}$ such that $\{\widetilde{\omega}_{j}\}_{j \geqslant \mathfrak{J}_{\bm k_*}}$ is covered with balls of size $2\mathfrak{J}_{\bm k_*}^{-\frac{1}{d}}$ centered at some element in $\{\widetilde{\omega}_{j}\}_{j \in \mathcal{I}_{\bm k_*}}$. 
	
	Then, for any $\bm \omega \in \mathfrak R (\bm k)$ with $\overline{\bm k} = j_0 > \max\left\{ \overline{\bm k_*}, \mathfrak{J}_{\bm k_*} \right\}$ and $\Pi (\bm k) = \bm k_*$,
	there exists $j_0' \in \mathcal{I}_{\bm k_*}$ 
	such that
	\begin{align}
		\nonumber
		|\omega_{j_0} - \omega_{j_0'}| 
		&\le
		|\widetilde{\omega}_{j_0} - \widetilde{\omega}_{j_0'}| 
		+ |\omega_{j_0} - \widetilde{\omega}_{j_0}| 
		+ |\omega_{j_0'} - \widetilde{\omega}_{j_0'}| 
		\\  \nonumber & 
		\le 2 \mathfrak{J}_{\bm k_*}^{-\frac{1}{d}}  
		+  j_0^{-\frac{1}{d}} +  j_0'^{-\frac{1}{d}}
		\\ & 
		\le 4 \mathfrak{J}_{\bm k_*}^{-\frac{1}{d}}.
		\label{omega'''}
	\end{align}
	Taking 
	\begin{align*}
		\bm k'=
		\begin{cases}
			k'_j = k_j, \quad j \le \overline{\bm k_*}, \\
			k'_j = 0, \quad   j > \overline{\bm k_*} \quad \text{and} \quad j \neq j_0',\\
			k'_{j_0'} = k_{j_0},
		\end{cases}
	\end{align*}
	we have
	\begin{align*}
		\overline{\bm k'}=j_0' \in \mathcal{I}_{\bm k_*} 
		\quad \text{and} \quad 
		\Pi\left(\bm k'\right) = \bm k_*.
	\end{align*}
	It  is worth noting that $\mathfrak{R}(\bm k) = \emptyset$ unless
	\begin{align}\label{kj0}
		|k_{j_0}| \le \frac{1}{a} + \frac{b}{a} \sum_{j \in \underline{\bm k}} |k_j|, 
		\quad \text{for} \quad j_0 = \overline{\bm k},
	\end{align}
	since 
	\begin{align}\label{ko}
		|\langle \bm k, \bm \omega \rangle| 
		\ge a |k_{j_0}| - b  \sum_{j \in \underline{\bm k}} |k_j|.
	\end{align}
	By using \eqref{Rk}, \eqref{omega'''} and \eqref{kj0}, we can conclude that for $\bm \omega$ given above,
	\begin{align*}
		|\langle \bm k', \bm \omega \rangle|
		&  
		\le |\langle \bm k, \bm \omega \rangle|
		+ \left| \omega_{j_0'} - \omega_{j_0} \right|  \left| k_{j_0}  \right|
		\\ & <
		\eta \mathfrak{B}_{\bm k_*}^{ C_0(d) }
		+
		4 \mathfrak{J}_{\bm k_*}^{-\frac{1}{d}} B_{\bm k_*}
		\\ &  <
		5 \eta \mathfrak{B}_{\bm k_*}^{ C_0(d) }
		,
		\quad
		\text{(by using \eqref{mfJ})}
	\end{align*}
	which implies $\bm \omega \in \widehat{\mathfrak R }(\bm k')$.
	Thus we finish the proof of this lemma.
\end{proof}

\subsection{Measure estimates}
Now we are in the position to prove Theorem \ref{thm:measure}.
\begin{proof}
	The proof is divided into two steps. 
	In the first step, we focus on $\operatorname{meas} \mathfrak R (\bm k)$.
	The estimate of $\operatorname{meas} \mathfrak R$ is completed in the second step, where Bourgain's Diophantine condition is employed to restrict the number of relevant indices $\bm k_*$.
	
	\noindent\textbf{Step 1:}
	In this step, we estimate $\operatorname{meas} \mathfrak R (\bm k)$.
	Specifically, we have the following estimate.
	\begin{lem}\label{Lem5.2}
		Fix $\bm k_* \in \mathbb{Z}^{\mathbb{N}}$.
		For $\bm k \in \mathbb{Z}^{\mathbb{N}}$ satisfying $\Pi(\bm k) = \bm k_*$, we have
		\begin{align*}
			\operatorname{meas} \mathfrak R (\bm k)
			\le 
			2 j_1^{\frac{1}{d}} 
			\eta \mathfrak{B}_{\bm k_*}^{ C_0(d) },
		\end{align*}
		where $j_1 = \overline{\bm k_*}$.
	\end{lem}
	\begin{proof}
		Since $0 \notin [a,b]$, we can conclude that for $\bm k \neq 0$ satisfying $\#(\operatorname{supp} \bm k) =1$, $\mathfrak R (\bm k) = \emptyset$ when $\eta$ is small enough.
		Thus we only need to consider the case $\#(\operatorname{supp} \bm k) \ge 2$, which implies that there exists $j \in \underline{\bm k}$ such that 
		\begin{align*}
			\frac{\partial }{\partial \omega_j} \langle \bm k, \bm\omega \rangle = k_j \neq 0.
		\end{align*}
		Then by using Fubini's theorem, we can conclude that
		\begin{align*}
			\operatorname{meas} \mathfrak R (\bm k) \le 
			2 j^{\frac{1}{d}} 
			\eta \mathfrak{B}_{\bm k_*}^{ C_0(d) }
			\le 
			2 j_1^{\frac{1}{d}} 
			\eta \mathfrak{B}_{\bm k_*}^{ C_0(d) },
		\end{align*}
		where $j_1 = \overline{\bm k_*}$.
	\end{proof}
	
	\noindent\textbf{Step 2:}
	Now we estimate $\operatorname{meas} \mathfrak{R}$.
	In view of \eqref{Rj0}, by using \eqref{kj0}, Lemma \ref{Lem5.1} and \ref{Lem5.2}, we can deduce that there exists a constant $C$ such that
	\begin{align*}
		\operatorname{meas} \mathfrak{R} 
		& \le 
		C
		\sum_{\bm k_* \in \mathbb{Z}^{\mathbb{N}} }
		\mathfrak{J}_{\bm k_*} 
		B_{\bm k_*}
		\eta \mathfrak{B}_{\bm k_*}^{ C_0(d) } 
		\\ & \le 
		C 
		\sum_{\bm k_* \in \mathbb{Z}^{\mathbb{N}} }
		B_{\bm k_*}^{1+d} 
		\eta^{1-d} \mathfrak{B}_{\bm k_*}^{(1-d)C_0(d)}
		j_0^{\frac{1}{d}} 
		\\ & \le 
		C(a,b,d)  \eta^{1-d}
		\sum_{\bm k_* \in \mathbb{Z}^{\mathbb{N}} }
		\mathfrak{B}_{\bm k_*}
		\\ & \le 
		C(a,b,d)  \eta^{1-d},
	\end{align*}
	where the last inequality is based on the fact
	\begin{align*}
		\sum_{\bm k_* \in \mathbb{Z}^{\mathbb{N}} }
		\mathfrak{B}_{\bm k_*} < \infty.
	\end{align*}
\end{proof}

\appendix 
\section{\hspace{-3mm} }\label{Ap}

\subsection{Proofs of Lemmas in Section~\ref{PB}}\label{app:proofs-section-2}

\subsubsection{Proof of Lemma \ref{Hamiltonian flow}}\label{A2}
\begin{proof}
	Expend $R \circ \Phi_F$ into Taylor series:
	\begin{align}\label{Hf1}
		R \circ \Phi_F = \sum_{n \ge 0} \frac{1}{n !} 
		R^{\langle n \rangle}, 
	\end{align}
	where $R^{\langle n +1 \rangle} = \left\{ R^{\langle n \rangle} , F \right\}, n \ge 0$ and $R^{\langle 0 \rangle}=R$.
	Note that, for $n \ge 1$, using Lemma \ref{lempb} $n$ times, we obtain
	\begin{align}
		\interleave R^{\langle n \rangle} \interleave_{\iota,r,h}
		& \le \nonumber
		\frac{24n}{r_0h_0}
		\interleave R^{\langle n -1 \rangle} \interleave_{\iota,r+\frac{r_0}{2n},h+\frac{h_0}{2n}}
		\interleave F \interleave_{\iota,r+\frac{r_0}{2},h+\frac{h_0}{2}}
		\\ \nonumber
		& \le 
		\interleave R^{\langle n -2 \rangle} \interleave_{\iota,r+\frac{r_0}{n},h+\frac{h_0}{n}}
		\left(
		\frac{24n}{r_0h_0}
		\interleave F \interleave_{\iota,r+\frac{n+1}{2n}r_0,h+\frac{n+1}{2n}h_0}
		\right)^2
		\\  \nonumber
		& \le  \cdots
		\\ 
		& \le  
		\interleave R \interleave_{\iota,r+r_0,h+h_0}
		\left(
		\frac{24n}{r_0h_0}
		\interleave F \interleave_{\iota,r+r_0,h+h_0}
		\right)^n.
		\label{Hf2}
	\end{align}
	Combining \eqref{Hf1} and \eqref{Hf2}, we have
	\begin{align*}
		\interleave R \circ \Phi_F \interleave_{\iota,r,h}
		& \le
		\interleave R \interleave_{\iota,r,h}
		+
		\sum_{n \ge 1} 
		\frac{1}{n!} \interleave R \interleave_{\iota,r+r_0,h+h_0}
		\left(
		\frac{24n}{r_0h_0}
		\interleave F \interleave_{\iota,r+r_0,h+h_0}
		\right)^n
		\\ & \le 
		\interleave R \interleave_{\iota,r,h}
		+
		\interleave R \interleave_{\iota,r+r_0,h+h_0}
		\sum_{n \ge 1} 
		\left(
		\frac{24e}{r_0h_0} 
		\interleave F \interleave_{\iota,r+r_0,h+h_0}
		\right)^n
		\\ & \quad (\text{ in view of } n^n < n! e^n)
		\\ & \le 
		\left( 1 +  \frac{48e}{r_0 h_0} 
		\interleave F \interleave_{\iota,r+r_0,h+h_0} \right)
		\interleave R \interleave_{\iota,r+r_0,h+h_0}
		,
	\end{align*}
	where the last inequality is based on \eqref{Hf}.
	Now we finish the proof of Lemma \ref{Hamiltonian flow}.
\end{proof}

\subsubsection{Proof of Lemma \ref{Hamiltonian vector field}}\label{A3}
\begin{proof}
	Given $j \in \mathbb{N}$, we write $F_{J_{j}}$ and $F_{\theta_{j}}$ as
	\begin{align*}
		F_{J_{\bm j}} 
		= \sum_{\bm n} F_{J_{j}, \bm n}
		= \sum_{\bm n}
		\frac{\partial }{\partial J_{j}}
		F_{\bm n}
	\end{align*}
	and
	\begin{align*}
		F_{\theta_{j}} 
		= \sum_{\bm n} 
		F_{\theta_{j}, \bm n}
		= \sum_{\bm n}
		\frac{\partial }{\partial \theta_{j}}
		F_{\bm n}.
	\end{align*}
	On the domain $D(r,h)$, we conclude that
	\begin{align}
		\left\|  F_J \right\|_{\infty} 
		= 
		\sup_{j \in \mathbb{N}} 
		\left| F_{J_{j}} \right| 
		&\le  \nonumber
		\sup_{j \in \mathbb{N}} 
		\sum_{\bm n   \atop  \operatorname{supp} \bm n \ni j} 
		\left\| F_{J_{j,\bm n}} \right\|_{r,h} 
	\\	& \le \nonumber
		\sup_{j \in \mathbb{N}} 
		\sum_{\bm n \atop  \operatorname{supp} \bm n \ni j }
		\left\| F_{\bm n} \right\|_{r+r_0,h} 
		{r_0}^{-1} 
		e^{w\ln^{\sigma}\lfloor j \rfloor}
		\\
		& \quad (\text{by the Cauchy estimate}) \nonumber
	\\	& \le \nonumber
		{r_0}^{-1}
		\sum_{j_0 \in \mathbb{N}} 
		e^{w\ln^{\sigma}\lfloor j_0 \rfloor}
		\sum_{\overline{\bm n} = j_0 }
		\left\| F_{\bm n} \right\|_{r+r_0,h} 
	\\	& \le 
		r_0^{-1}
		\interleave F \interleave_{\iota,r+r_0,h}.
		\label{FI}
	\end{align}
	Similarly, on the domain $D(r,h)$ we have
	\begin{align}
		\left\|  F_{\theta} \right\|_{w} 
		&= \nonumber
		\sum_{j \in \mathbb{N}} 
		\left| F_{\theta_{j}} \right|
		e^{w\ln^{\sigma}\lfloor j \rfloor}
	\\ & \le  \nonumber
		\sum_{j \in \mathbb{N}} 
		\sum_{\bm n \atop  \operatorname{supp} \bm n \ni j}
		|k_{j}|
		\| F_{\bm n} \|_{r,h}
		e^{w\ln^{\sigma}\lfloor j \rfloor}
	\\ & \le  \nonumber 
		\sum_{j_0 \in \mathbb{N}}
		e^{w\ln^{\sigma}\lfloor j_0 \rfloor}
		\sum_{\overline{\bm n} = j_0 }
		\sum_{j \in \mathbb{N}} 
		|k_{j}|
		\| F_{\bm n} \|_{r,h}
	\\	& \le 
		\frac{1}{h_0}
		\interleave F \interleave_{\iota,r,h+h_0},
		\label{Ft}
	\end{align}
	where we have used the fact
	\begin{align*}
		|\bm k| = \sum_{\bm j \in \mathbb{N}} |k_{\bm j}|.
	\end{align*}
In view of Definition~\eqref{DefHV}, the estimates \eqref{FI} and \eqref{Ft} complete the proof of Lemma~\ref{Hamiltonian vector field}.
\end{proof}

\section*{Acknowledgments}

This work is supported by NNSFC (No.12571195, No.12371241) and Natural Science Foundation of Liaoning Province (No.2025-MS-002).

\section*{Data Availability}
The manuscript has no associated data.

\section*{Declarations}
{\bf Conflicts of interest} \ The authors  state  that there is no conflict of interest.

\end{document}